\documentclass[a4paper,envcountsame]{llncs}
\pdfoutput=1
\newcommand{\mypresubsubspace}{\vspace{-6pt}}

\usepackage{xspace}
\usepackage{latexsym}
\usepackage{amssymb}   %
\usepackage{color}

\usepackage{wrapfig}

\usepackage{pgf}
\usepackage{hyperref}
\hypersetup{bookmarksdepth=2, bookmarksopenlevel=2}
\usepackage[hyphenbreaks]{breakurl}

\newtheorem{df}[theorem]{Definition}
\newtheorem{propo}[theorem]{Proposition}
\newcommand*{\seq}[2][n]
            {{\ensuremath{#2_{1}, \allowbreak \ldots, \allowbreak #2_{#1}}}}
\newcommand*{\SEQ}[3]
            {{\ensuremath{#1_{#2}, \allowbreak \ldots, \allowbreak #1_{#3}}}}
\newcommand*{\SEQC}[3]
            {{\ensuremath{#1_{#2} \cdots #1_{#3}}}}

\newcommand*{\partto}{\hookrightarrow}
\newcommand{\restrict}[2]{#1\,{\rule[-1.15ex]{.3pt}{3ex}}\,_{ #2}}

\newcommand*{\vars}{{\ensuremath{\it vars}}}

\newcommand*{\mydash}{{\mbox{\tt-}}}
\newcommand*{\HU}{{\ensuremath{\cal{H U}}}\xspace}
\newcommand*{\HB}{{\ensuremath{\cal{H B}}}\xspace}
\newcommand*{\TU}{{\ensuremath{\cal{T U}}}\xspace}
\newcommand*{\TB}{{\ensuremath{\cal{T B}}}\xspace}
\newcommand*{\M}{{\ensuremath{\cal M}}\xspace}

\usepackage[mathscr]{euscript} 
\renewcommand*{\S}{{\ensuremath{\mathscr S}}\xspace}

\newcommand*{\NN}{{\ensuremath{\mathbb{N}}}\xspace}

\DeclareRobustCommand{\mydate}
  {\ifthenelse{\boolean{commentsaon}}{\makebox[0pt][l]{\quad Draft, \today}}{}}

\author{W{\l}odzimierz Drabent
}

\title{On completeness of logic programs%
\thanks{Accepted, without the appendix,
     for post-conference proceedings of LOPSTR 2014
    (24th International Symposium on Logic-Based Program Synthesis and
    Transformation. Canterbury, UK, September 2014), to appear in 2015 
    in Springer LNCS.
}
} %
\institute{Institute of Computer Science,
         Polish Academy of Sciences,
         and
	\\
         IDA, Link\"opings universitet, Sweden
 \\     {\tt drabent\,{\it at}\/\,ipipan\,{\it dot}\/\,waw\,{\it dot}\/\,pl}
 \\
[2.5ex]November 10, 2014\vspace*{-7ex}
  }
\begin{document}

\maketitle

\begin{abstract}
Program correctness (in imperative and functional programming) splits in
logic programming into correctness and completeness.  Completeness means that
a program produces all the answers required by its specification.  Little
work has been devoted to reasoning about completeness.  This paper presents a
few sufficient conditions for completeness of definite programs.  We also
study preserving completeness under some cases of pruning of SLD-trees
(e.g. due to using the cut).

We treat logic programming as a declarative paradigm, abstracting from any
operational semantics as far as possible.  We argue that the proposed methods
are simple enough to be applied, possibly at an informal level, in practical
Prolog programming.
We point out importance of approximate specifications.

 \vspace{-2.5ex}
\end{abstract}


\section{Introduction}

The notion of partial program correctness splits in logic programming into
correctness and completeness.
Correctness means that all answers of the program are compatible with the specification,
completeness -- that the program  produces all the answers required by the
specification. 

In this paper we consider definite clause programs, and present a few
sufficient conditions for their completeness.  We also discuss preserving
completeness under pruning of SLD-trees (by e.g.\ using the cut).
We are interested in declarative reasoning, i.e.\ abstracting from any
operational semantics, and treating program clauses as logical formulae.
Our goal is simple methods, which may be applied -- possibly informally -- in
actual practical programming.

\paragraph{Related work.}
Surprisingly little work was devoted to proving completeness of programs.
Hogger \cite{hogger.book} defines the notion of completeness, but does not
provide any sufficient conditions.  
Completeness is not discussed in the important monograph
\cite{Apt-Prolog}.
Instead, a characterization is studied of the set of computed instances of an
atomic query, in a special case when the set is finite and the answers are
ground. 
In the paper 
\cite{kowalski85shorter} of Kowalski
completeness is discussed, but the example proofs concern
only correctness.  
As a sufficient condition for completeness of a program $P$ 
he suggests
$P\vdash T_S$, where $T_S$ is a specification in a form of a logical theory.
The condition seems
impractical as it fails when $T_S$ contains auxiliary
predicates, not 
occurring in $P$. 
It also requires that all the models of $P$ (including the Herbrand base)
are models of the specification.  
But it seems that such specifications often have a
substantially restricted class of models, maybe a single Herbrand
model, cf.\ \cite{Deville}.

Deville \cite{Deville} provides an approach where correctness and completeness
of programs should follow from construction.
No direct sufficient criteria for completeness, applicable to arbitrary programs, are given.
Also the approach is not 
declarative, as it is based on an operational semantics of SLDNF-resolution.

St\"{a}rk
\cite{DBLP:journals/jlp/Stark98}
presents an elegant method of reasoning about a broad class of properties of
programs with negation,
executed under LDNF-resolutions.  
A tool to verify proofs mechanically was provided.
The approach involves a rather complicated induction scheme, so 
it seems impossible to apply the method informally by programmers.
Also, the approach is not fully declarative, as the order of literals in
clause bodies is important.

A declarative sufficient condition for program completeness
was given by Deransart and Ma{\l}uszy\'nski \cite{Deransart.Maluszynski93}.
The approach presented here stems from
\cite{DBLP:journals/tplp/DrabentM05shorter},
the differences are discussed in
the full version of this paper  \cite{drabent.report2014}.
The main contribution since the former version \cite{drabent12.iclp.shorter}
is proving completeness of pruned SLD-trees.
The author is not aware of any other work on this issue.

\paragraph{Preliminaries.}
We use the standard notation and definitions \cite{Apt-Prolog}.
An atom whose predicate symbol is $p$ will be called
a {\em$p$-atom} (or an {\em atom for $p$}).  Similarly, a clause
whose head is a $p$-atom is a {\em clause for $p$}.  
In a program $P$, 
by {\em procedure $p$} we mean the set of the clauses for $p$ in $P$.

We assume a fixed alphabet with an infinite set of function symbols.
The Herbrand universe will be denoted by \HU, the Herbrand base by \HB,
and the sets of all terms, respectively atoms, by \TU and \TB.
For an expression (a program) $E$
by $ground(E)$ we mean the set of ground instances
of $E$ (ground instances of the clauses of $E$).
$\M_P$ denotes the least Herbrand model of a program $P$.

By ``declarative'' (property, reasoning, \ldots) 
we mean referring only to logical reading of programs, 
thus abstracting from any operational semantics.
In particular, 
properties depending on the order of atoms in clauses will not be considered
declarative (as they treat equivalent conjunctions differently).

  By a computed (respectively correct) answer for a program $P$ and a
  query $Q$ we mean an instance $Q\theta$ of $Q$ where $\theta$ is a computed
  (correct) answer substitution \cite{Apt-Prolog} for $Q$ and $P$.
  We often say just
   {\em answer}
 as each computed answer is a correct one, 
  and each correct answer (for $Q$) is a computed answer
  (for $Q$ or for some its instance).
  Thus, by soundness and completeness of SLD-resolution,
  $Q\theta$ is an answer for $P$ iff $P\models Q\theta$.

  Names of variables begin with an upper-case letter.
  We use the list notation of Prolog.  So 
  $[\seq t]$  ($n\geq0$) stands for the list of elements $\seq t$.
  Only a term of this form is considered a list.
 (Thus terms like $[a,a|X]$, or  $[a,a|a]$, where $a$ is a constant,
  are not lists).
The set of natural numbers will be denoted by \NN;
 $f\colon A\partto B$  states that $f$ is a partial function from $A$ to $B$.

The next section introduces the basic notions of specifications, correctness
and completeness. 
Also, advantages of approximate specifications are discussed.
After a brief overview of proving correctness, 
we discuss proving program completeness.
Sect.\,\ref{sec:pruning} deals with proving that completeness is preserved under
pruning.  
We finish with a discussion.
For missing proofs, more examples etc 
see \cite{drabent.report2014}.

\section{Correctness and completeness}

\subsubsection{Specifications.}
The purpose of a logic program is to compute a relation, or a few relations.
A specification should describe these relations.  It is convenient to assume
that the relations are over the Herbrand universe.  
To describe such relations, one relation corresponding to each procedure of the
program (i.e.\ to a predicate symbol),
it is convenient to use a Herbrand interpretation.
Thus a (formal) {\bf specification} is a Herbrand interpretation, i.e.\ a
subset of \HB.

\mypresubsubspace
\subsubsection{Correctness and completeness.}
In imperative and functional programming, correctness usually means that the
program results 
are as specified.  In logic programming, due to its non-deterministic nature,
we actually have two issues: {\em correctness} (all the results are
compatible with the specification) and {\em completeness} (all the results
required by the specification are produced). 
In other words, correctness means that the relations defined by the program are
subsets of the specified ones, and completeness means inclusion in the
opposite  direction. 
In terms of specifications and the least Herbrand models
we define:
{\sloppy\par}
\begin{df}
\label{def:corr:compl}
Let $P$ be a program and $S\subseteq\HB$ a specification.
$P$ is {\bf correct} w.r.t.\ $S$ when $\M_P\subseteq S$;
it is {\bf complete} w.r.t.\ $S$ when $\M_P\supseteq S$.
\end{df}
We will sometimes skip the specification when it is clear from the context.
We propose to call a program {\bf fully correct} when it is both correct and complete.
If a program $P$ is fully correct  w.r.t.\ a specification $S$ then,
obviously, $\M_P=S$.

A program $P$ is correct w.r.t.\ a specification $S$ iff
$Q$ being an answer of $P$  implies $S\models Q$.
(Remember that $Q$ is an answer of $P$ iff $P\models Q$.)
The program is complete w.r.t.\ $S$ iff
$S\models Q$ implies that $Q$ is an answer of $P$.
(Here our assumption on an infinite set of function symbols is needed
\cite{drabent.report2014}.)

It is sometimes useful to consider local versions of these notions:

\begin{df}
\label{def:corr:compl:local}
A {\bf predicate} $p$ in $P$ is
 {\bf correct} w.r.t.\ $S$ when each $p$-atom of $\M_P$ is in $S$, and 
 {\bf complete} w.r.t.\ $S$ when each $p$-atom of $S$ is in $\M_P$.

An {\bf answer} $Q$ is  {\bf correct} w.r.t.\ $S$ when $S\models Q$.

$P$ is {\bf complete for a query} $Q$ w.r.t.\  $S$
when
$S\models Q\theta$ implies that $Q\theta$ is an answer for $P$,
for any ground instance $Q\theta$ of $Q$.

\end{df}
Informally,  $P$ is complete for $Q$
when 
all the answers for $Q$ required by the specification $S$ are answers of $P$.
Note that a program is complete w.r.t.\ $S$
 iff it is complete w.r.t.\ $S$ for any query
 iff it is complete w.r.t.\ $S$ for any query $A\in S$.

\mypresubsubspace
\subsubsection{Approximate specifications.}
\label{par:approximate}
Often it is difficult, and not necessary, to specify the relations defined by
a program exactly;  more formally, to require that $\M_P$ is equal to a given
specification.  Often the relations defined by programs are not exactly those
intended by programmers.  For instance this concerns the programs in
Chapter 3.2
of the textbook \cite{Sterling-Shapiro-shorter}
defining predicates
{\tt member/2},  {\tt append/3}, {\tt sublist/2}, and some others.
The defined relations are not those of list membership, concatenation,
etc.
However this is not an error, as for all intended queries the answers are 
as for a program defining the intended relations.
The exact semantics
of the programs
is not explained in the textbook;  such explanation is not needed.
Let us look more closely at {\tt append/3}.
\begin{example}
1.\mbox{ }%
The program  APPEND
\[
  app(\,[H|K],L,[H|M]\,) \gets app(\,K,L,M\,). \qquad\qquad
\linebreak[3]
  app(\,[\,],L,L\,). \
\]  
does not define the relation of  list concatenation.  
For instance, ${\rm APPEND} \models app([\,],1,1)$.
In other words, APPEND is not correct w.r.t.\ 
\[
S_{\rm APPEND}^0 = \{\, app(k,l,m)\in\HB \mid
                k,l,m \mbox{ are lists, } k*l=m
\,\},
\] 
where $k*l$ stands for the concatenation of lists $k,l$.
It is however complete w.r.t.\ $S_{\rm APPEND}^0$, and correct w.r.t.\ 
\[
S_{\rm APPEND} = \{\, app(k,l,m)\in\HB \mid
                \mbox{if $l$ or $m$ is a list then }
                app(k,l,m)\in S_{\rm APPEND}^0
\,\}.
\] 
Correctness w.r.t.\,$S_{\rm APPEND}$ and completeness 
w.r.t.\,$S_{\rm APPEND}^0$ are sufficient to show that
APPEND will produce the required 
results when used to concatenate or split lists.
More precisely, the answers for a query  $Q=app(s,t,u)$, where $t$
 is a list or $u$ is a list,
are $app(s\theta, t\theta, u\theta)$, where 
$s\theta, t\theta, u\theta$ are lists and $s\theta * t\theta = u\theta$.
(The lists may be non-ground.)

2.\mbox{ }%
Similarly, the procedures {\tt member/2} and {\tt sublist/2} are complete
w.r.t\ specifications describing the relation of list membership, and the
sublist relation.  It is easy to provide specifications, w.r.t.\ which
the procedures are correct.  For instance, 
{\tt member/2} is correct w.r.t.\ 
$
S_{\rm MEMBER}= \{\, member(t,u )\in\HB \mid
                \mbox{if }\linebreak[3]
                u=[\seq t] \mbox{ for some } n\geq0 \mbox{ then }
                t=t_i,  \mbox{ for some }
                 0<i\leq n \,\}
$.

3.\mbox{ }%
The exact relations defined by programs are often misunderstood.
For instance, in
\cite[Ex.\,15]{DBLP:journals/jlp/DevilleL94shorter}
it is claimed that a program $P r o g_1$
defines the relation of list inclusion.  In our terms, this means that 
predicate 
${\it included}$ of  $P r o g_1$ is correct and complete w.r.t.\ 
\[
\left\{\, {\it included}(l_1,l_2)\in\HB\ \left| \
\begin{array}{l}
     l_1,l_2 \mbox{ are lists},
\\
\mbox{every element of $l_1$ belongs to $l_2$}  
\end{array}
\right.\right\}.
\]
However the correctness does not hold:
The program contains a unary clause ${\it included}([\,],L)$,
so $P r o g_1\models{\it included}([\,],t)$ for any term $t$,
\end{example}

The examples show that in many cases it is unnecessary to know the semantics
of a program exactly.  Instead it is sufficient to describe it approximately.
An {\bf approximate specification} is a pair of specifications
$S_{\it c o m p l}, S_{corr}$, for completeness and correctness.
The intention is that the program is complete w.r.t.\ the former, and correct
w.r.t.\ the latter:
$S_{\it c o m p l}\subseteq \M_P \subseteq S_{corr}$.
In other words, the specifications $S_{\it c o m p l}, S_{corr}$ describe,
 respectively, which atoms have to be computed, and
which are allowed to be computed.
For the atoms from 
$S_{corr} \setminus  S_{\it c o m p l}$
the semantics of the program is irrelevant.
By abuse of terminology, $S_{corr}$ or  $S_{\it c o m p l}$ 
will sometimes also be called
approximate specifications.

\mypresubsubspace
\subsubsection{Proving correctness}
\label{sec:correctness}
We briefly discuss proving correctness, as it is complementary to the 
main subject of this paper.
The approach is due to Clark~\cite {Clark79}.

\begin{theorem}[Correctness]
\label{th:correctness}
A sufficient condition for a program $P$ to be correct w.r.t.\ 
a specification $S$ is
\[
\begin{tabular}l
     for each ground instance 
    $
    H\gets \seq B
    $
    of a clause of the program,
    \\
     if $\seq B\in S$ then $H\in S$.
\end{tabular}
\]
\end{theorem}

\begin{example}
\label{ex:split-corr}%
Consider a program SPLIT and a specification describing
how the sizes of the last two arguments of $s$ are related
($|l|$ denotes the length of a list~$l$):
\newcommand*{\specsplit}{\ensuremath{S}\xspace}
\begin{eqnarray}
  &&     s( [\,], [\,], [\,] ).     \label{split1}\\   
  &&     s( [X|X s], [X|Y s], Zs ) \gets s( X s, Z s, Y s ).   \label{split2}
\\\nonumber
\specsplit &=& \{\, s(l, l_1, l_2) \mid 
                 l, l_1, l_2 \mbox{ are lists, } 0\leq |l_1|-|l_2| \leq 1 \,\}.
\end{eqnarray}
SPLIT is correct w.r.t.\ \specsplit, by Th.\,\ref{th:correctness}
(the details are left for the reader, or see
\cite{drabent.report2014}).  
A stronger specification for which SPLIT is correct is shown in 
Ex.\,\ref{ex:split-compl}.
\end{example}

The sufficient condition is equivalent to
$S\models P$, and to $T_P(S)\subseteq S$.

Notice that the proof method is declarative.
The method should be well known, but is often neglected.
For instance it is not mentioned in \cite{Apt-Prolog}, where a more
complicated method, moreover not declarative, is advocated.
That method is not more powerful than the one of Th.\,\ref{th:correctness}
 \cite {DBLP:journals/tplp/DrabentM05shorter}.
See \cite {DBLP:journals/tplp/DrabentM05shorter,drabent.report2014}
for further examples, explanations, references and discussion.

\section{Proving completeness}
\label{sec:completeness}
We first introduce a notion of semi-completeness, and sufficient conditions
under which semi-completeness of a program implies its completeness.
Then a sufficient condition follows for semi-completeness.
We conclude the section with a way of showing completeness directly without employing semi-completeness.

\begin{df}
A {\em level mapping} is a
function $|\ |\colon \HB\to \NN$ assigning natural numbers to atoms.

A program $P$ is  {\bf recurrent} {w.r.t.\ a level mapping}~$|\ |$
\cite{Apt-Prolog} if, in
every ground instance  $H\gets\seq B\in ground(P)$ of its clause ($n\geq0$),
$|H|>|B_i|$ for all $i=1,\ldots,n$.
A program is {\em recurrent}
if it is recurrent w.r.t.\ some level mapping.   

\label{def:acceptable}%

A program $P$ is {\bf acceptable} w.r.t.\ a specification $S$ and a level
mapping $|\ |$ if 
$P$ is correct w.r.t.\ $S$, and for every
$H\gets\seq B\in ground(P)$
we have $|H|>|B_i|$ whenever $S\models B_1,\ldots,B_{i-1}$.
A program is {\em acceptable} if it is acceptable w.r.t.\ some level mapping
and some specification.
\end{df}

The definition of acceptable is more general than that of 
\cite{Apt-Prolog}, 
which requires $S$ to be a model of $P$.
 Both definitions make the same programs acceptable \cite{drabent.report2014}.

\begin{df}
A program $P$ is {\bf semi-complete} 
w.r.t.\ a specification $S$ if
$P$ is complete w.r.t.\ $S$ for any query $Q$ for which there exists a finite SLD-tree.
\end{df}

Less formally, the existence of a finite SLD-tree means
that $P$ with $Q$ terminates under some selection rule.
For a semi-complete program, if a computation for a query $Q$
terminates then all the required by the specification answers for $Q$ have
been obtained.
Note that a complete program is semi-complete.
Also:

\begin{propo}
[Completeness]
\label{prop:semi-compl}%
Let a program $P$ be semi-complete w.r.t.\ $S$. 
The program is complete w.r.t\ $S$ if
\begin{enumerate}
\item 
\label{prop:semi-compl:term}
for each query $A\in S$ there exists a finite SLD-tree, or

each $A\in S$ is an instance of a query $Q$ for which a finite SLD-tree
exists,~or
\item 
\label{prop:semi-compl:recu}
the program is recurrent, or
\item 
\label{prop:semi-compl:accept}
 the program is acceptable
(w.r.t.\ a specification $S'$ possibly distinct from $S$).
\end{enumerate}
\end{propo}
\subsubsection
{Proving semi-completeness.}
We need the following notion.

\begin{df}
  A ground atom $H$ is
  {\bf covered by a clause} $C$ w.r.t.\ a specification $S$ \cite{Shapiro.book}
  if $H$ is the head of a ground instance  
  $
  H\gets \seq B
  $
  ($n\geq0$) of $C$, such that all the atoms $\seq B$ are in $S$.
  A ground atom $H$ is {\bf covered} {\bf by a program} $P$ w.r.t.\ $S$
  if it is covered w.r.t.\ $S$ by some clause $C\in P$.
\end{df}

For instance, given a specification 
$S = \{ p(s^i(0))\mid i\geq0 \}$,
atom $p(s(0))$ is covered both by
$p(s(X))\gets p(X)$ and by
$p(X)\gets p(s(X))$.

Now we present a sufficient condition for semi-completeness.
Together with  Prop.\,\ref{prop:semi-compl} it
provides a sufficient condition for completeness.

\begin{theorem}[Semi-completeness]
\label{th:completeness}%
If all the atoms from a specification $S$ are covered w.r.t.~$S$
 by a program $P$ 
then $P$ is semi-complete w.r.t.~$S$.
\end{theorem}

\begin{example}
\label{ex:split-compl}%
We show that program SPLIT from Ex.\,\ref{ex:split-corr} is complete w.r.t.
\[
S_{\rm SPLIT} = 
\left\{\,
\begin{array}{l}
 s([\seq[2n]t], [t_1,\cdots,t_{2n-1}], [t_2,\cdots,t_{2n}] ),
\\
 s([\seq[2n+1]t], [t_1,\cdots,t_{2n+1}], [t_2,\cdots,t_{2n}] )
\end{array}
   \,
   \left|
   \,
   \begin{array}l
      n\geq0,\\ \seq[2n+1] t\in \HU  \\
   \end{array}
\right.\right\}
,
\]
where
$[t_k,\cdots,t_l]$ denotes the list $[t_k,t_{k+2},\ldots,t_l]$,
for $k,l$  both odd or both even.

Atom
$s([\,],[\,],[\,])\in S_{\rm SPLIT}$ is covered by clause (\ref{split1}).
For $n>0$,
any atom
$A= \linebreak[3]
s([\seq[2n]t], \linebreak[3]
      [t_1,\cdots,t_{2n-1}],\linebreak[3] [t_2,\cdots,t_{2n}] )$
is covered by an instance of  (\ref{split2}) with a body
$B= \linebreak[3]
     s([t_2,\ldots,t_{2n}], \linebreak[3]
     [t_2,\cdots,t_{2n}] ,\linebreak[3] [t_3,\cdots,t_{2n-1}]
 )$.
Similarly, 
for $n\geq0$ and any atom
$A= s([\seq[2n+1]t],\linebreak[3]
         [t_1,\cdots,t_{2n+1}],\linebreak[3] [t_2,\cdots,t_{2n}] )$,
the corresponding body is
$B= s([t_2,\ldots,t_{2n+1}],\linebreak[3] [t_2,\cdots,t_{2n}],
 \linebreak[3]     [t_3,\cdots,t_{2n+1}] )$.
In both cases, $B\in S_{\rm SPLIT}$.
(To see this, rename each $t_i$ as $t'_{i-1}$.)
So $S_{\rm SPLIT}$ is covered by SPLIT.
Thus SPLIT is semi-complete w.r.t.\ $S_{\rm SPLIT}$, by
Th.\,\ref{th:completeness}.

Now by Prop.\,\ref{prop:semi-compl} the program is complete, as it is
recurrent under the level mapping
 $  |  s(t,t_1,t_2) | = |t| $, where
 $  |\, [h|t]\, | = 1+|t| $ and 
 $  |f(\seq t)| = 0 $
 (for any ground terms $h,t,\seq t$, and any function symbol $f$ 
  distinct from $[\ | \ ]$\,).

By Th.\,\ref{th:correctness}
the program is also correct w.r.t.\ $S_{\rm SPLIT}$, as 
$S_{\rm SPLIT}\models {\rm SPLIT}$.  (The details are left to the reader.)
Hence $S_{\rm SPLIT}= \M_{\rm SPLIT}$.
\end{example}

    Note that the sufficient condition of Th.\,\ref{th:completeness}
     is equivalent to $S\subseteq T_P(S)$,
    which implies $S\subseteq {\rm g f p}(T_P)$.
It is also equivalent to $S$ being a model of ONLY-IF$(P)$
(see e.g.\ \cite{Doets} or \cite{DBLP:journals/tplp/DrabentM05shorter} for a definition).

The notion of semi-completeness is tailored for finite programs.  An SLD-tree
for a query $Q$ and an infinite program $P$ may be infinite, but with all
branches finite.  In such case,
if the condition of Th.\,\ref{th:completeness} holds
then $P$ is complete for $Q$ \cite{drabent.report2014}.

\subsubsection
{Proving completeness directly.}
\label{sec:completeness-directly}
Here we present another declarative way of proving completeness;
a condition is added to Th.\,\ref{th:completeness} so that completeness is
implied directly.  This also works for non-terminating programs.
However when termination has to be shown anyway, applying
Th.\,\ref{th:completeness} seems simpler.

In this section we allow that a level mapping is a {\em partial}
function $|\ |\colon \HB\partto \NN$ assigning natural numbers to some atoms.

\begin{df}
  A ground atom $H$ is {\bf recurrently covered}
  by a program $P$   w.r.t.\ a specification $S$ and a level mapping 
$|\ |\colon \HB\partto \NN$
  if $H$ is the head of a ground instance  
  $
  H\gets \seq B
  $
($n\geq0$)
of a clause of the program, such that
$|H|, |B_1|, \ldots |B_n|$ are defined, 
$\seq B\in S$,
and $|H|>|B_i|$ for all $i=1,\ldots,n$.

\end{df}

For instance, given a specification 
$S = \{\, p(s^i(0))\mid i\geq0 \,\}$,
atom $p(s(0))$ is recurrently covered by a program
$\{\, p(s(X))\gets p(X). \}$ under a level mapping for which 
$|p(s^i(0))|=i$.
No atom is recurrently covered
by \mbox{$\{\, p(X)\gets p(X). \}$}.
Obviously, if $H$ is recurrently covered by $P$ then it is covered by $P$.
If $H$ is covered by $P$ w.r.t.\ $S$ and $P$ is recurrent w.r.t. $|\;|$ 
then $H$ is recurrently covered w.r.t. $S,|\;|$.  The same holds for $P$
acceptable w.r.t.\ an $S'\supseteq S$.

\begin{theorem}[Completeness 2]
\label{th:completeness:recu}%
 {\rm (A reformulation of Th.\,6.1 of \cite{Deransart.Maluszynski93}).}
If, under some level mapping $|\ |\colon \HB\partto \NN$,
 all the atoms from a specification $S$ are
recurrently covered by a program $P$ w.r.t.\ $S$ 
then $P$ is complete w.r.t.\ $S$.
\end{theorem}

\begin{example}
\label{ex:infinite}%
Consider a directed graph $E$.  As a specification for a program describing
reachability in $E$, take 
$S = S_p\cup S_e$,
where
\[
\begin{tabular}{l}
$
S_p = \{\, p(t,u) \mid
 \mbox{there is a path}\linebreak[3]\ 
 \mbox{from } t \mbox{ to } u \mbox{ in } E \,\}
$,
\\
$
S_e = \{\, e(t,u) \mid (t,u)\mbox{ is an edge in } E \,\}
$.  
\end{tabular}
\]
Let $P$ consist of a procedure $p$: 
$
\{\,    p(X,X). \ \
    p(X,Z) \gets e(X,Y),\, p(Y,Z). 
\}
$
and a procedure $e$ which is a set of unary clauses describing the edges of
the graph.  Assume the latter is complete w.r.t.\ $S_e$.
Notice that when $E$ has cycles then 
infinite SLD-trees cannot be avoided, and
completeness of $P$ cannot be shown by 
Prop.\,\ref{prop:semi-compl}.

 To apply Th.\,\ref{th:completeness:recu}, let
 us define a level mapping for the elements of $S$ such that
$|e(t,u)| = 0$ and
$|p(t,u)|$ is the length of a shortest path in $E$ from $t$ to $u$
 (so $|p(t,t)|=0$).
Consider a $p(t,u)\in S$ where $t\neq u$.
Let 
$t=t_0,\seq t=u$ be a shortest path from $t$ to $u$.
Then $e(t,t_1),p(t_1,u)\in S$,
$|p(t,u)|=n$, $|e(t,t_1)|=0$, and $|p(t_1,u)|=n-1$.
Thus $p(t,u)$ is recurrently covered by $P$ w.r.t.\ $S$ and $|\ |$.
The same trivially holds for the remaining atoms of $S$.
So $P$ is complete w.r.t.~$S$.
\end{example}

\section{Pruning SLD-trees and completeness}
\label{sec:pruning}

Pruning some parts of SLD-trees is often used to improve efficiency of programs.
It is implemented by using the cut,  the if-then-else construct of
 Prolog, or  built-ins, like {\tt once/1}. 
Pruning preserves the correctness of a logic program, it also preserves
termination under a given selection rule,
but may violate the program's completeness.  
We now discuss proving
 that completeness is preserved.

By a {\bf pruned SLD-tree} for a program $P$ and a query $Q$ we mean a tree
with the root $Q$
which is a connected subgraph of an SLD-tree for $P$ and $Q$.
By an  {\em answer} of a pruned SLD-tree we mean the computed answer of a
successful SLD-derivation which is a branch of the tree.
We will say that a pruned SLD-tree $T$ with root $Q$
is {\bf complete} w.r.t.\ a specification $S$
if, for any ground  $Q\theta$,
$S\models Q\theta$ implies that $Q\theta$ is an instance of an answer of $T$.
Informally, such a tree produces all the answers for $Q$ required by $S$.

We present two approaches for proving completeness of pruned SLD-trees.
The first one is based on viewing pruning as skipping certain clauses
while building the children of a node.  The other deals with a restricted usage
of the cut.

  \newcommand{\myellipse}{%
      \begin{pgfpicture}{-1.5cm}{-.5cm}{1.5cm}{.5cm}
        \pgfsetxvec{\pgfpoint{.5cm}{0cm}}
        \pgfsetyvec{\pgfpoint{0cm}{.5cm}}
        \pgfsetlinewidth{1pt}
        \pgfmoveto{\pgfxy(-3,0)}
        \pgfcurveto{\pgfxy(-3,-1)}{\pgfxy(3,-1)}{\pgfxy(3,0)}
        \pgfstroke
        \pgfmoveto{\pgfxy(3,0)}
        \pgfsetdash{{3pt}{3pt}}{0pt}
        \pgfcurveto{\pgfxy(3,1)}{\pgfxy(-3,1)}{\pgfxy(-3,0)}
        \pgfstroke
      \end{pgfpicture}
  }
  \newcommand{\dblue}{\color[rgb]{0,0,.5}}
  \newcommand{\dgreen}{\color[rgb]{0,0.5,0}}

  \newcommand{\mydiagram}{%
  \begin{pgfpicture}{-2.5cm}{-.5cm}{2.5cm}{1.7cm}
        \pgfsetxvec{\pgfpoint{.5cm}{0cm}}
        \pgfsetyvec{\pgfpoint{0cm}{.5cm}}
    \pgfputat{\pgfxy(0,3)}{\pgfbox[center,center]
      {{\footnotesize\ldots,}$\underline A${\footnotesize,\ldots}}
    }
    \pgfline{\pgfxy(-0.6,2.5)}{\pgfxy(-5,0)}
    \pgfline{\pgfxy(-0.2,2.5)}{\pgfxy(-.5,0)}
    \begin{pgfscope}
      \pgfsetdash{{3pt}{1.5pt}}{0pt}
      \pgfline{\pgfxy(0.2,2.5)}{\pgfxy(0.5,0)}
      \pgfline{\pgfxy(0.6,2.5)}{\pgfxy(5,0)}
    \end{pgfscope}

    \pgfputat{\pgfxy(-4.9,.7)}
             {\pgfbox[center,center]{\mbox{\small\footnotesize$\dgreen\Pi_{i}$}}}
    \pgfputat{\pgfxy(-1.1,1)}{
      \pgfbox[center,center]{\dgreen
        \begin{pgfmagnify}{.7}{.6}
            \myellipse
        \end{pgfmagnify}
        }
    }
    \pgfputat{\pgfxy(-1.2,1.2)}{{\pgfbox[center,center]{$\cdots$}}}
    \pgfputat{\pgfxy(+2.2,0.6)}{{\pgfbox[center,center]{$\cdots$}}}

    \pgfputat{\pgfxy(-2.9,1.9)}
             {\pgfbox[center,center]{\mbox{\small\footnotesize$\dblue P$}}}
    \pgfputat{\pgfxy(-2.4,1.3)}{\dblue
        \begin{pgfmagnify}{.8}{.5}
            \myellipse
        \end{pgfmagnify}
    }      
    \pgfputat{\pgfxy(2.5,-.5)}
             {\pgfbox[center,center]{\mbox{\small\footnotesize pruned}}}  
    \pgfputat{\pgfxy(-2.5,-.5)}
             {\pgfbox[center,center]{\mbox{\small\footnotesize not pruned}}}  
  \end{pgfpicture}%
 } %

{ %

\subsection{Pruning as clause selection.}

To facilitate reasoning about the answers of pruned SLD-trees,
we will now view pruning as applying only certain clauses 
while constructing the children of a 
\setlength{\intextsep}{0.4ex}
\begin{wrapfigure}{r}{4.5cm}
\hfill
\scalebox{.7}
    {\mydiagram}
\hfill\mbox{}
\end{wrapfigure}
node. 
So we introduce subsets $\seq\Pi$ of $P$.  The intention is that
for each node the clauses of exactly one $\Pi_i$ are used.
Programs $\seq\Pi$ may be not disjoint.

} %

\vspace{.4ex plus 2ex}

\begin{df}
Given programs $\seq{\Pi}$ ($n>0$),
a {\bf c-selection rule} is a function%
\
assigning to a query $Q'$ an atom $A$
in $Q'$ and one of the programs $\emptyset,\seq{\Pi}$.

A {\bf csSLD-tree} (cs for clause selection) for a query
$Q$ and programs
$\seq\Pi$, via a c-selection rule $R$,
is constructed as an SLD-tree, but
for each node its children are constructed using the program selected by the
c-selection rule.
An answer of a csSLD-tree is defined in the expected way.  
\sloppy
\end{df}

A c-selection rule may choose the empty program, thus 
making a given node a leaf.
Notice that a csSLD-tree for $Q$ and $\seq\Pi$ is a pruned SLD-tree
for $Q$ and $\bigcup_i\Pi_i$.   Conversely, for each pruned SLD-tree $T$
for $Q$ and a (finite) program $P$ there exist $n>0$, and
$\seq\Pi\subseteq P$ such that 
 $T$ is a csSLD-tree for $Q$ and $\seq\Pi$.

\begin{example}
\label{ex:prune}%
 We show that completeness of each of $\seq\Pi$ is not sufficient for
 completeness of a csSLD-tree for $\seq\Pi$.
    Consider a program $P$:
\vspace*{-.5\abovedisplayskip}
\par\noindent\mbox{}\hfill
\hspace*{-3em}
 \begin{minipage}{.35\textwidth}%
    \begin{eqnarray}
      &&
\label{Excl1}
    q(X)\gets p(Y,X). \\
      &&
\label{Excl2}
     p(Y,0).
    \end{eqnarray}
  \end{minipage}
\qquad
  \begin{minipage}{.4\textwidth}
    \begin{eqnarray}  
\label{Excl3}  &&
     p(a,s(X))\gets p(a,X).      \\
\label{Excl4} &&
      p(b,s(X))\gets p(b,X).
    \end{eqnarray}
  \end{minipage}
\hfill\mbox{}
\vspace{\belowdisplayskip}
\par\noindent
and programs $\Pi_1= \{(\ref{Excl1}), (\ref{Excl2}), (\ref{Excl4})\}$,
$\Pi_2= \{(\ref{Excl1}), (\ref{Excl2}), (\ref{Excl3})\}$,
    As a specification for completeness consider
    $S_0=\{\,q(s^j(0)) \mid j\geq0 \,\}$.
  Each of the programs  $\Pi_1, \Pi_2,P$ is complete w.r.t.\ $S_0$.
    Assume a c-selection rule $R$ choosing alternatively $\Pi_1,\Pi_2$ along
    each branch of a tree. 
    Then the csSLD-tree for $q(s^j(0))\in S_0$ via $R$ (where $j>2$) 
    has no answers,
    thus the tree is not complete w.r.t.\ $S_0$.

\end{example}

Consider programs $P,\seq\Pi$ and specifications $S,\seq S$, such that 
$P\supseteq\bigcup_{i=1}^n\Pi_i$ and $S=\bigcup_{i=1}^n S_i$.
The intention is that each $S_i$ describes which answers are to be produced by
using $\Pi_i$ in the first resolution step.
We will call $\seq {\Pi}$, $\seq S$ a {\bf split} (of $P$ and $S$).
Note that  $\seq {\Pi}$ or $\seq S$ may be not disjoint.

\begin{df}
\label{def:suitable}
Let  $\S =\seq {\Pi}$, $\seq S$ be a split, and $S=\bigcup S_i$.

Specification $S_i$ is {\bf suitable} for an atom $A$
w.r.t.\ \S
when no instance of $A$ is in $S\setminus S_i$.
(In other words, when $ground(A)\cap S \subseteq S_i$.)
We also say that a program $\Pi_i$ is {\bf suitable} for $A$ w.r.t.\ \S
when $S_i$ is.

A c-selection rule is {\bf compatible} with \S if for each non-empty query
$Q$ it selects an atom $A$
and a program $\Pi$, such that 

\quad
-- $\Pi\in\{\seq\Pi\}$ is suitable for $A$ w.r.t.\ \S, or

\quad
-- none of $\seq\Pi$ is suitable for $A$ w.r.t.\ \S and  $\Pi=\emptyset$
(so $Q$ is a leaf).

A csSLD-tree for $\seq {\Pi}$ via a c-selection rule compatible with \S
is said to be {\bf weakly compatible} with \S.
The tree is {\bf compatible} with \S iff
for each its nonempty node some $\Pi_i$ is selected.
\end{df}

  The intuition is that 
  when $\Pi_i$ is suitable for $A$
  then $S_i$ is a fragment of $S$ sufficient to deal with $A$.
  It describes all the answers for query $A$ required by~$S$.

The reason of incompleteness of the trees in Ex.\,\ref{ex:prune}
may be understood 
as selecting a $\Pi_i$ not suitable for the selected atom.
Take  $\S=\Pi_1,\Pi_2, S_0\cup S_1',S_0\cup S_2'$,
where
$S_1' = \{\,p(b,s^i(0)) \mid i\geq0 \,\}$ and
$S_2' = \{\,p(a,s^i(0)) \mid i\geq0 \,\}$.
In the incomplete trees,
$\Pi_1$ is selected for an atom $A=p(a,u)$,
or $\Pi_2$ is selected for an atom $B=p(b,u)$ (where $u\in\TU$).
However
$\Pi_1$ is not suitable for $A$ 
whenever $A$ has an instance in $S$
(as then
 $ground(A)\cap S \not\subseteq S_0\cup S_1'$);
similarly for $\Pi_2$ and~$B$.

  When  $\Pi_i$ is suitable for  $A$ then if
    each atom of $S_i$ is covered by $\Pi_i$
  (w.r.t.\ $S$)
  then using for $A$ only the clauses of $\Pi_i$ does not impair
  completeness w.r.t.~$S$:

\newcommand{\thcompletenesspruned}[1]
{%
    Let $P\supseteq\bigcup_{i=1}^n\Pi_i$ (where $n>0$) be a program, 
    $S=\bigcup_{i=1}^n S_i$ a specification, and 
    $T\!$ a csSLD-tree for $\seq\Pi$.
    If
    \begin{enumerate}
    \item 
    \label{prop:cssld.complete.cond0}
        for each $i=1,\ldots,n$, 
        all the atoms from $S_i$ are covered by $\Pi_i$ w.r.t.\ $S$, and 
    \item
    \label{prop:cssld.complete.cond1}
        $T$ is compatible with $\seq\Pi,\seq S$,
    \item 
    \label{prop:cssld.complete.cond2}
    \begin{enumerate}
      \item 
      \label{prop:cssld.complete.cond2a}
      #1
      \item 
      \label{prop:cssld.complete.cond2b}
          program $P$ is recurrent, or
      \item 
      \label{prop:cssld.complete.cond2c}
        $P$ is acceptable
        (possibly w.r.t.\ a specification distinct from $S$), and
          $T$ is built under the Prolog selection rule 
    \end{enumerate}
    \end{enumerate}
    then $T$ is complete w.r.t.\ $S$.
} %

\begin{theorem}
\label{th:completeness:pruned}%
\label{prop:cssld.complete}%
\thcompletenesspruned{         $T$ is finite, or}
\end{theorem}
\begin{example}
\label{ex:pruning1}%
The following program SAT0 is a simplification of a fragment of the SAT solver 
of \cite{howe.king.tcs}
discussed in \cite{drabent12.iclp.shorter}.
Pruning is crucial for the efficiency and usability of the original program.
\begin{quote}
\vspace{-2.5\abovedisplayskip}%
\mbox{}\hspace{-2.5em}
\mbox{%
  \begin{minipage}[t]{.53\textwidth}%
    \begin{eqnarray}
      &&
\label{SCcl1}
      p( P\mydash P, \, [\,] ).
          \\
      &&
\label{SCcl2}
      p( V\mydash P, \, [B|T] ) \gets q( V\mydash P, \, [B|T] ).   \\
      &&
\label{SCcl3}
      p( V\mydash P, \, [B|T] ) \gets q( B, \, [V\mydash P|T] ).   
    \end{eqnarray}
  \end{minipage}
  \begin{minipage}[t]{.43\textwidth}%
    \begin{eqnarray}
      &&
\label{SCcl4}
      q( V\mydash P, \, \mbox{\LARGE\_}\, )\gets V=P.                 \\
      &&
\label{SCcl5}
      q(\, \mbox{\LARGE\_}\, , \, [A|T] )\gets p( A, T).              \\
      &&
\label{SCcl6}
      P = P.
    \end{eqnarray}
  \end{minipage}  
} %
\end{quote}
The program is complete w.r.t.\ a specification
\[
S =
\left.\left
\{\, 
\begin{array}{l}
  p(t_0\mydash u_0,[t_1\mydash u_1,\ldots, t_n\mydash u_n]),      \\
  q(t_0\mydash u_0,[t_1\mydash u_1,\ldots, t_n\mydash u_n])    
\end{array}
\, \right| \,
\begin{array}{l}
    n\geq0,
    \ t_0,\ldots,t_n,u_0,\ldots,u_n \in\mathbb{T},
    \\
    t_i=u_i \mbox{ for some } i\in \{0,\ldots,n\}\,
\end{array}
\right
\}\cup S_=
\]
where $\mathbb T = \{{\it false}, {\it true}\}\subseteq\HU$, and
$S_= = \{\, t{=}t \mid t\in\HU \,\}$.
We omit a completeness proof,
mentioning only that SAT0 is recurrent w.r.t.\ a level mapping
$|p(t,u)| = 2|u|+2$, 
$|q(t,u)| = 2|u|+1$, $|{=}(t,u)|=0$,
where $|u|$ is as in Ex.\,\ref{ex:split-compl}.

The first case of pruning is due to redundancy within 
(\ref{SCcl2}), (\ref{SCcl3});
 both
$\Pi_1={\rm SAT0}\setminus \{(\ref{SCcl3})\}$ and 
$\Pi_2={\rm SAT0}\setminus \{(\ref{SCcl2})\}$
are complete w.r.t.\ $S$.
For any selected atom at most one of (\ref{SCcl2}), (\ref{SCcl3}) is to be
used, and the choice is dynamic.
As the following reasoning is independent from this choice, we omit further
explanations.

So in 
such pruned SLD-trees
the children of each node are constructed using one of programs
$\Pi_1, \Pi_2$.
Thus they are csSLD-trees for $\Pi_1, \Pi_2$.
They are compatible with $\S= \Pi_1, \Pi_2, S, S$
(as $\Pi_1, \Pi_2$ are trivially suitable for any $A$, due to 
$S_i=S$ and $S\setminus S_i=\emptyset$ in Df.\,\ref{def:suitable}).
Each atom of $S$ is covered w.r.t.\ $S$ both by $\Pi_1$ and $\Pi_2$.
As SAT0 is recurrent,
by Th.\,\ref{th:completeness:pruned}, 
each such tree is complete w.r.t.\ $S$.

\end{example}

\begin{example}
\label{ex:pruning2}%
We continue with program SAT0 and specification $S$ from the previous
example, and add
 a second case of pruning.
When the selected atom is of the form $A=q(s_1,s_2)$ with a ground $s_1$ then
only one of 
clauses (\ref{SCcl4}), (\ref{SCcl5}) is needed --
(\ref{SCcl4}) when $s_1$ is of the form $t\mydash t$, and  (\ref{SCcl5})
otherwise.
The other clause can be abandoned without losing the completeness w.r.t.\ $S$.%
\footnote
{
 The same holds for $A$ of the form  $q(s_{11}\mydash s_{11},s_2)$, or
 $q(s_{11}\mydash s_{12},s_2)$ with non-unifiable $s_{11}$, $s_{12}$.
 The pruning is implemented using the if-then-else construct in
 Prolog: \
{\tt q(V-P,[A|T]) :-  V=P ->  true ; p(A,T).}
(And the first case of pruning by
{\tt p(V-P,[B|T]) :- non{}var(V) -> q(V-P,[B|T])  ; q(B,[V-P|T])}.)
}%

 Actually, SAT0 is included in a bigger program, say
 $P={\rm SAT0}\hspace{.5pt}\cup\Pi_0$.
 We skip the details of $\Pi_0$, let us only state that $P$ is recurrent,
 $\Pi_0$ does not
 contain any clause for $p$ or for $q$, and that
 $P$ is complete w.r.t.\ a specification $S'=S\cup S_0$ where
 $S_0$ does not contain any $p$- or $q$-atom.
 (Hence each atom of $S_0$ is covered by $\Pi_0$ w.r.t.\ $S'$.)

To formally describe the trees for $P$ resulting from both cases of pruning,
consider $\S = \Pi_0,\ldots,\Pi_5,S_0,\ldots,S_5$, where
\[
\begin{tabular}[t]{l @{\ \ }l}
   $\Pi_1 =  \{(\ref{SCcl1}),(\ref{SCcl2})\}$, \
    $\Pi_2 =  \{(\ref{SCcl1}),(\ref{SCcl3})\}$,
&       $S_1=S_2=S\cap \{\,p(s,u)\mid s,u\in\HU\,\}$,
\\
  $\Pi_3 =  \{(\ref{SCcl4})\}$,
&       $S_3=S\cap\{\,q(t\mydash t,s)\mid t,s\in\HU \,\}$,
\\
   $\Pi_4 =  \{(\ref{SCcl5})\}$,
&       $S_4= S\cap\{\,q(t\mydash u,s)\mid  t,u,s\in\HU,t\neq u \,\}$,
\\
   $\Pi_5 =  \{(\ref{SCcl6})\}$,
&  $S_5 = S_=$.
\end{tabular}
\]
 Each atom from $S_i$ is covered by $\Pi_i$ w.r.t.\ $S$ (for $i=0,\ldots,5$).
For each $q$-atom with its first argument ground,
$\Pi_3$ or $\Pi_4$ (or both) is suitable.
For each remaining atom from \TB,
a program from $\Pi_0,\Pi_1,\Pi_2,\Pi_5$ is suitable.

Consider a pruned SLD-tree $T$ for $P$ (employing the two cases of pruning
described above).  
Assume that each $q$-atom selected in $T$ has its first argument ground.
Then $T$ is a csSLD-tree compatible with \S.
From Th.\,\ref{th:completeness:pruned} it follows that $T$ is complete
w.r.t.\ $S$.

The restriction on the selected $q$-atoms can be
implemented by means of Prolog delays.  This can be done in such a way 
that floundering is avoided for the intended initial queries
\cite{howe.king.tcs}.
So the obtained pruned trees are as $T$ above, and
 the pruning preserves completeness of the program.
\end{example}

\begin{example}
A Prolog program
$
\{
n o p(a d a m,0) \gets {!}.
\ \
n o p(eve,0)  \gets {!}.
\ \
n o p(X,2).
\}
$
is an example of difficulties and dangers of using the cut in Prolog.
Due to the cut, for an atomic query $A$
only the first clause with the head unifiable with $A$ will
be used.
The program can be seen as logic program $P=\Pi_1\cup\Pi_2\cup\Pi_3$ executed
with pruning, where (for $i=1,2,3$)
$\Pi_i$ is the $i$-th clause of the program
with the cut removed.
The intended meaning is $S=S_1\cup S_2\cup S_3$, where 
$S_1 = \{ n o p(a d a m,0) \}$, 
$S_2 = \{ n o p(eve,0) \}$, and
$S_3 = \left\{ \rule{0pt}{1.9ex}
     n o p(t,2)\in\HB \mid t\not\in\{a d a m,eve\} \right\}$.
Note that all the atoms from $S_i$ are covered by $\Pi_i$ (for $i=1,2,3$).
(We do not discuss here the (in)correctness of the program.)

Let \S be $\Pi_1,\Pi_2,\Pi_3,S_1, S_2, S_3$.  
Consider a query $A=n o p(t,Y)$ with a ground $t$.
If $t=a d a m$ then only $\Pi_1$ is suitable for $A$ w.r.t.\ \S,
if $t=eve$ then only  $\Pi_2$ is.
For $t\not\in\{a d a m,eve\}$ the suitable program is $\Pi_3$.
So for a query $A$ the pruning due to the cuts in the program results in 
selecting a suitable $\Pi_i$, and the obtained
csSLD-tree compatible with \S. 
By Th.\,\ref{th:completeness:pruned}
the tree is complete w.r.t.\ $S$.

For a query $n o p(X,Y)$ or $n o p(X,0)$ only the first clause, i.e.\ $\Pi_1$,
is used.  However $\Pi_1$ is not suitable for the query (w.r.t.\ \S),
and the csSLD-tree is not compatible with \S.
The tree is not complete (w.r.t.\ S).
\end{example}

\subsection{The cut in the last clause.}

The previous approach is based on a somehow abstract semantics
in which pruning is viewed as clause selection.
Now we present an approach referring directly to Prolog with the cut.
However the usage of the cut is restricted to the last clause of a procedure.
The author expects that the general case could be conveniently studied 
in the context of programs with negation
(because if $H\gets A_1,!,A_2$ is followed by $H\gets A_3$ then the latter
clause is used only if $A_1$ fails).
We consider LD-resolution, as interaction of the cut
with delays introduces additional complications.

We need to reason about the atoms selected in the derivations. 
So we employ 
a (non-declarative) approach
to reason about LD-derivations, presented in \cite{Apt-Prolog}.
A specification in this approach,
let us call it {\bf call-success specification}, is a pair $pre,post\in\TB$
of sets of atoms, closed under substitution.
A program is {correct} w.r.t.\ such specification,
let us say {\bf c-s-correct}, when 
in each LD-derivation
every selected atom is from $pre$
and each corresponding computed answer is in $post$,
provided that the derivation begins with an atomic query from $pre$.
See \cite{Apt-Prolog} or \cite{DBLP:journals/tplp/DrabentM05shorter}
for further explanations, and for a sufficient criterion for c-s-correctness.

By $\vars(E)$ we denote the set of variables occurring in an expression $E$.
For a substitution $\theta = \{X_1/t_1,\ldots,X_n/t_n\}$, 
let 
$do m(\theta)=  \{X_1,\ldots,X_n\}$, and $r n g(\theta)=\vars(\{\seq t\})$.

\begin{df}\rm
\label{def:adjustably}
Let $S$ be a specification, and  $pre,post$ a call-success specification.
  A ground atom $A$ is {\bf adjustably covered} by a clause $C$
w.r.t.\ $S$ and $pre,post$
 if $A$ is covered by $C$ and the cut does not occur in $C$, or 
the following three conditions hold:
\begin{enumerate}
\item 
$C$ is $H\gets \seq[k-1]A,!,A_k,\ldots,A_n$, 

\item
$A$ is covered by $H\gets A_1,\ldots,A_{k-1}$ w.r.t.\ $S$,

\item
\label{def:adjustably:tough:condition}
\begin{itemize}
\item 
 for any instance $H\rho\in pre$  such that $A$ is an instance of $H\rho$,

\item 
for any ground instance  $(A_1,\ldots,A_{k-1})\rho\eta$ such that
 $ A_1\rho\eta,\ldots,A_{k-1}\rho\eta\in post$,

\item
$A$ is covered by 
 $(H \gets A_k,\ldots,A_n)\rho\eta$ w.r.t.\ $S$,

\end{itemize}

where
 $do m(\rho)\subseteq \vars(H)$, $r n g(\rho) \cap \vars(C) \subseteq \vars(H)$, 
 $do m(\rho)\cap r n g(\rho) = \emptyset$,
 and $do m(\eta) = \vars((A_1,\ldots,A_{k-1})\rho)$.
{\sloppy\par}
\end{enumerate}
\end{df}

Informally, condition \ref{def:adjustably:tough:condition} says that 
$A$ could be produced out of each ``related'' answer for $\seq[{k-1}]A$,
and some answer for $A_k,\ldots,A_n$ specified by $S$.
Note that 
if $A$ is adjustably covered by $C$ w.r.t.\  $S$, $pre,post$,
where  $S\subseteq post$,
then $A$ is covered by $C$ w.r.t.\  $S$.
\label{th:completeness:cut:simplification}
If condition \ref{def:adjustably:tough:condition} holds for an $H\rho$ then it
holds for any its instance $H\rho\delta$ of which $A$ is an instance.
For a proof of this property and of the theorem below, see Appendix.

\begin{theorem}\rm
\label{th:completeness:cut}
Let $S$  be a specification, $pre,post$ a call-success specification,
where $S\subseteq post$.
Let $T$ be a pruned LD-tree for a program $P$ and an atomic query $Q$, where pruning is
due to the cut occurring in the last clause(s) of some procedure(s) of $P$.
If 

 -- \  $T$ is finite, $Q\in pre$, $P$ is c-s-correct w.r.t.\ $pre,post$, and 

 -- \  each $A\in S$ is adjustably covered by a clause of $P$ w.r.t.\  $S$ and
    $pre,post$
 \\
then $T$ is complete w.r.t.\ $S$.
\end{theorem}

\begin{example}
Consider a program  IN  and specifications:
\[
\begin{array}{l}
    \begin{array}[b]{l}
      in([\,],L). \\
      in([H|T],L) \gets m(H,L), !, in( T, L).
    \end{array}
    \qquad
    \begin{array}[b]{l}
      m(E,[E|L]). \\
      m(E,[H|L]) \gets m(E,L).
    \end{array}
\end{array}
\]
\[
\begin{array}{l}
 S = S_m\cup S_{in},\ \  pre=pre_m\cup pre_{in}, \ \
post=post_m\cup post_{in}, \mbox{ \ where}
\\[.5ex]
pre_{m}= \{\, m(u,t)\in\TB \mid t \mbox{ is a list} \,\},
\\
pre_{in}= \{\, in(u,t)\in\HB \mid u,t \mbox{ are ground lists} \,\},
\\
post_{m}= \{\, m(t_i, [\seq t] )\in\TB \mid 1\leq i\leq n \,\},
\\
post_{in}= \{\, in([\seq[m]u], [\seq t] )\in\HB  \mid
            \{\seq[m]u\}\subseteq\{\seq t\} \,\},
\\
S_m = post_m\cap\HB, \ S_{in} = post_{in}.
\end{array}
\]
The program is c-s-correct w.r.t.\ $pre,post$ (we skip a proof).
We show that each atom 
 $A = in( u,t )\in S_{in}$, where $u=[\seq[m]u]$, $m>0$,
 is adjustably covered 
by the second clause $C$ of IN.  Let $C_0$ be $in([H|T],L) \gets m(H,L)$.
Now $A$ is covered by $C_0$ w.r.t.\ $S$ 
($A\gets m(u_1,t)$ is a relevant ground instance of $C_0$).

Take an instance $in([H|T],L)\rho\in pre$ of the head of $C$.
The instance is ground, and the whole $C\rho$ is ground.  
So in Df.\,\ref{def:adjustably}, $\rho\eta=\rho$.
If  $A$ is an instance of (thus equal to) $in([H|T],L)\rho$
then $in(T,L)\rho = in([u_2,\ldots,u_m],t)\in S$ (as $A\in S$).
Thus $A$ is covered by $(in([H|T],L)\gets in(T,L))\rho$.

Thus $A$ is adjustably covered by $C$.  It is easy to check that 
all the remaining atoms of $S$ are covered by IN w.r.t.\ $S$, 
and that IN is recurrent
(for $|m(s,t)|=|t|$, $|in(s,t)|=|s|+|t|$, $|t|$ as in Ex.\,\ref{ex:split-compl}).
Thus each LD-tree for IN and a query $Q\in pre$ is finite.
By Th.\,\ref{th:completeness:cut}, each such tree pruned due to the cut
is complete w.r.t.\ S.
Notice that condition \ref{def:adjustably:tough:condition} does not hold
when non ground arguments of $in$ are allowed in $pre_{in}$, 
and that for
such queries some answers may be pruned.

\end{example}

\section{Discussion}

\paragraph{Declarativeness.}

Without declarative ways of reasoning about correctness and completeness of
programs, logic programming 
would not deserve to be called a declarative programming paradigm.
The sufficient condition for proving correctness (Th.\,\ref{th:correctness}),
that for semi-completeness of Th.\,\ref{th:completeness}, 
and those for completeness of  
Prop.\,\ref{prop:semi-compl}.\ref{prop:semi-compl:recu} 
and Th.\,\ref{th:completeness:recu} are declarative.
However 
the sufficient conditions for completeness of
 Prop.\,\ref{prop:semi-compl}.\ref{prop:semi-compl:term} and
 \ref{prop:semi-compl}.\ref{prop:semi-compl:accept} are not, 
as they refer to program termination, or depend on the order of atoms in 
clause bodies.  

 Declarative completeness proofs employing 
Prop.\,\ref{prop:semi-compl}.\ref{prop:semi-compl:recu} 
or Th.\,\ref{th:completeness:recu}
imply termination, or require reasoning similar to that in termination proofs.
So proving completeness by means of semi-completeness and termination 
may be a reasonable compromise
between declarative and non-declarative reasoning,
as termination has to be shown anyway in most of practical cases.

\paragraph{Granularity of proofs.}
Note that the sufficient condition for correctness deals with single clauses,
that for semi-completeness -- with procedures, and those for completeness
take into account a whole program.

\paragraph{Incompleteness diagnosis.}
There is a close relation between completeness proving and
incompleteness diagnosis \cite{Shapiro.book}.  As the reason of
incompleteness, a diagnosis algorithm finds an atom from $S$ that is not
covered by the program.  Thus it finds a reason for violating the sufficient
conditions for semi-completeness and completeness of
 Th.\,\ref{th:completeness}, \ref{th:completeness:recu}.

\paragraph{Approximate specifications.}
We found that approximate specifications are crucial in avoiding unnecessary
complications in dealing with correctness and completeness of programs
(cf.\ Sect.\,\ref{par:approximate}, 
\cite{DBLP:journals/tplp/DrabentM05shorter,%
      drabent12.iclp.shorter,drabent.report2014}).
For instance, in the main example of \cite{drabent12.iclp.shorter}
(and in its simpler version in Ex.\,\ref{ex:pruning1},\,\ref{ex:pruning2})
finding an exact specification is not easy, and is unnecessary.  
The required property of the program is described more conveniently by an
approximate specification.
Moreover, as this example shows, in program development the
semantics of (common predicates in) the consecutive versions of a program may
differ.  What is unchanged is correctness and completeness w.r.t.\ an
approximate specification.

{\em Approximate specifications in program development.}
This suggests a generalization of the paradigm of
program development by semantics preserving program transformations
 \cite{DBLP:journals/jlp/PettorossiP94}:
it is useful and natural to use transformations which only preserve correctness
and completeness w.r.t.\ an approximate specification.
{\em Approximate specifications in debugging.}
In declarative diagnosis \cite{Shapiro.book}
the programmer is required to know the exact intended semantics of the
program.  This is a substantial obstacle to using declarative diagnosis
in practice.
Instead, an approximate specification can be used,
 with the specification for correctness
(respectively completeness) 
applied in incorrectness (incompleteness)
diagnosis.  See  \cite{drabent.report2014} for discussion and references.

\paragraph{Interpretations as specifications.}
\label{par:specifications:interpretations}
This work uses specifications which are interpretations.
(The same kind of specifications is used, among others,
in \cite{Apt-Prolog}, and in declarative diagnosis.)
There are however properties which cannot be expressed by
such specifications \cite{DBLP:journals/tplp/DrabentM05shorter}.  
 For instance one cannot express that some instance of an
atomic query $A$ should be an answer; one has to specify the actual instance(s).
Other approach is needed for such properties, possibly with specifications
which are logical theories (where axioms like $\exists X.\, A$ can be used).

\paragraph{Applications.}
\label{par:applications}
We want to stress the simplicity and naturalness of the sufficient conditions
for correctness (Th.\,\ref{th:correctness}) and semi-completeness 
(Th.\,\ref{th:completeness},
the condition
 is a part of each discussed sufficient condition for completeness).
Informally, the first one says that the clauses of a program should produce
only correct conclusions, given correct premises.  
  The other says that each ground atom that should be produced by $P$ 
  can be produced by a clause of $P$
  out of atoms produced by $P$.
The author believes that this is a way a competent programmer reasons about 
(the declarative semantics of) a logic program. 

Paper \cite{drabent12.iclp.shorter} illustrates practical applicability of the
methods presented here.
It shows a systematic construction of a non-trivial Prolog program
(the SAT solver of \cite{howe.king.tcs}).  
Starting from a formal specification, a definite clause logic program
is constructed hand in hand with proofs of its correctness, completeness,
and termination under any selection rule.
The final Prolog program is obtained by adding control to the logic program
(delays and pruning SLD-trees).
Adding control preserves correctness and termination.
However completeness may be
violated by pruning, and by floundering related to delays.
By Th.\,\ref{prop:cssld.complete}, 
the program with pruning remains complete.%
\footnote{%
In \cite{drabent12.iclp.shorter}
 a weaker version of Th.\,\ref{prop:cssld.complete}
was used, and one case of pruning was discussed informally. 
A proof covering both cases of pruning is illustrated here in
Ex.\,\ref{ex:pruning2}. 
}
Proving non-floundering is outside of the scope of this work.
See \cite{king.non-suspension2008} for a related analysis algorithm, 
 applicable in this case \cite{king.personal}.

The example shows how well
   ``logic'' could be separated from ``control.''  
The whole reasoning related to correctness and completeness can be done 
declaratively, abstracting from any operational semantics.

\paragraph{Future work.}
A natural continuation is developing completeness proof methods for programs
with negation
(a first step was made in \cite{DBLP:journals/tplp/DrabentM05shorter}),
maybe also for constraint logic programming and CHR (constraint handling
rules). 
Further examples of proofs are necessary.
An interesting task is formalizing and automatizing the proofs, a first step
is formalization of specifications.

\subsubsection{Conclusion.}
Reasoning about completeness of logic program has been, surprisingly,
almost neglected.
This paper presents a few sufficient conditions for completeness.
As an intermediate step we introduced a notion of semi-completeness.
The presented methods are, to a large extent, declarative.
Examples suggest that the approach is applicable -- maybe at informal level --
in practice of Prolog programming. 
The approach is augmented by
 two methods of proving completeness in presence of pruning.

\appendix
\section*{Appendix}
\label{sec:appendix}
\pdfbookmark[1]{Appendix}{bookmark:appendix}

The appendix contains a proof of 
Th.\,\ref{th:completeness:cut} and of a property stated on
p.\,\pageref{th:completeness:cut:simplification}.  We begin with an
additional example of applying Th.\,\ref{th:completeness:cut}.

   \begin{example}
   Consider a program $P$:
   \[
   p(X,Z) \gets q(X,Y), !, r(Y,Z).
   \qquad\quad
   \begin{array}[t]{l}
     q(a,a) \\ q(a,a') \\ q(b,b) 
   \end{array}
   \qquad\quad
   \begin{array}[t]{l}
     r(a,c) \\   r(a',c) 
   \end{array}
   \vspace{-1\belowdisplayskip}
   \]
   and specifications
   \[
   \begin{array}{l}
       S=\{\,p(a,c),  q(a,a'),q(b,b),   r(a,c),   r(a',c) \,\}, \\
       post = S\cup\{q(a,a)\}, \\
       pre = \{\, p(a,t) \mid t\in\TU \,\} \cup
        \{\, q(a,t) \mid t\in\TU \,\} \cup
        \{\, r(t,u) \mid t,u\in\TU \,\}
   \end{array}
   \]
   The program is c-s-correct w.r.t.\ $pre,post$ (we skip a proof).
   To check that 
   atom $p(a,c)\in S$ is adjustably covered by the first clause of $P$, 
   note first that it is covered w.r.t.\ $S$ by 
   $p(a,c)\gets q(a,a')$.
   It is sufficient to check 
   condition \ref{def:adjustably:tough:condition} of Df.\,\ref{def:adjustably}
   for $\rho=\{X/a\}$,
   as $p(X,Z)\rho=p(a,Z)$ is a most general $p$-atom in $pre$
   (cf.\ Lemma \ref{lemma:instances:adjustably:tough:condition} below).
   If  $q(X,Y)\rho\eta \in post$ then $\eta=\{Y/a\}$ or $\eta=\{Y/a'\}$.
   Hence $r(Y,Z)\rho\eta$ is $r(a,Z)$ or $r(a',Z)$.
   In both cases, 
   $p(a,c)\gets r(Y\eta,c)$ is a ground instance of
   $(p(X,Z)\gets r(Y,Z))\rho\eta$ covering $p(a,c)$ w.r.t.\ $S$.

   The remaining atoms of $S$ are trivially covered by the unary clauses of $P$.
   The LD-tree for $P$ and $Q=p(a,Z)$ is finite, hence the LD-tree pruned due to
   the cut is complete w.r.t.\ $S$ by Th.\,\ref{th:completeness:cut}.

   \end{example}

Before the proof of Th.\,\ref{th:completeness:cut}.
let us discuss how the cut works.
We treat Prolog programs as definite programs, the same for queries.
The cut is considered as additional control information.
We however often write symbol  ! in queries, to remind the original position
of the cut in a program clause.

Assume a Prolog program $P$, which is a logic program in which, additionally,
the cut may occur in the last clause of a procedure.
Consider a pruned LD-tree $T$ resulting from pruning an LD-tree
$T_0$ due to the cut.
The cut is involved whenever a query $Q_{i-1}$ has a child $Q_i$:
\[
\begin{array}{l}
Q_{i-1}=A,Q', \\  
Q_i = (\seq[k-1]A,!,A_k,\ldots,A_n,Q')\theta_i,  \\
\end{array}
\]
where $H\gets \seq[k-1]A,!,A_k,\ldots,A_n$ is the clause variant used and 
$\theta_i$ is an mgu of $A$ and $H$.
Note first that the cut introduced in $Q_i$ may affect only the subtree of
$T_0$ rooted in  $Q_i$ (as the clause with the cut is the last in its
procedure). 
The top part of the subtree of $T_0$ rooted in  $Q_i$ can be seen as the
LD-tree $T'$ for $(\seq[k-1]A)\theta_i$
 (with an instance of $A_k,\ldots,A_n,Q'$ added to each query of $T'$).
If  $T'$ contains no success then no pruning is performed due to the cut in
$Q_i$.
Also, no pruning happens when there is an infinite branch in $T'$
preceding all the success branches.
Otherwise, pruning is performed and all the successes, but one, are pruned
away.
More precisely, exactly one path remains not pruned, out of all the paths in
$T_0$ beginning in $Q_i$  of the form  
\[
\begin{array}{l}
Q_i =(\seq[k-1]A,!,A_k,\ldots,A_n,Q')\theta_i,
 \\
Q_{i+1} =
(\ldots,A_2,\ldots,A_{k-1},!,A_k,\ldots,A_n,Q')\theta_i\theta_{i+1},  \\
\cdots \\
Q_{j-1} = (A',!,A_k,\ldots,A_n,Q')\theta_i\cdots\theta_{j-1},  
\\
Q_j = (A_k,\ldots,A_n,Q')\theta_i\cdots\theta_{j},  
\end{array}
\]
(where each query $Q_i,\ldots,Q_{j-1}$ contains more atoms than $Q_j$ does).

Strictly speaking, it was assumed here that no cut is introduced in any query
$Q_l$ ($i<l<j$).
To deal with such extra cuts, notice that the same reasoning applies recursively
(i.e.\ by induction on the number of cuts introduced within $Q_i,\ldots,Q_j$).
So we showed that:
\begin{equation}
\label{property:pruned:tree}
 \parbox{.83\textwidth}
 {\color{black}
    If the LD-tree $T'$ for $(\seq[k-1]A)\theta_i$ contains a successful branch, 
    not preceded by an infinite branch,
    then the pruned tree $T$ contains a path
    $Q_{i-1},\ Q_i,\ \ldots, \ (A_k,\ldots,A_n,Q')\theta_i\cdots\theta_{j}$.

    Otherwise no pruning occurs due to the cut introduced in $Q_i$.
    (The cut is not executed. No success leaf is a descendant of $Q_i$ in $T$.) 
 }
\end{equation}

\begin{proof}
 [of Th.\,\ref{th:completeness:cut}]
\ 
As each atom of $S$ is covered by $P$ w.r.t.\ $S$, 
$P$ is semi-complete w.r.t.\ $S$  by Th.\,\ref{th:completeness}.
Consider the LD-tree $T_0$ for $P$ and $Q$, and the finite pruned LD-tree $T$.
Without loss of generality we can assume that the employed mgu's are
idempotent and relevant
\cite[p.\,37--38]{Apt-Prolog}.

Consider a ground instance $Q\theta$ of $Q$ such that $S\models Q\theta$.
In the proof of  Th.\,\ref{th:completeness}
(cf.\ \cite[Th.\,4]{drabent.corr.2012}),
 out of a ground successful derivation $D$ for
$Q\theta$ a lift \cite[Df.\,5.35]{Doets} was constructed, which was a branch
of the tree for $Q$. 
Each atom occurring in $D$ was from $S$.
Here such a ground derivation may not exist. 
Instead we construct a lift (for a superset of $P$) which consists of some nodes
of a successful branch of $T$.   
Roughly speaking, a fragment of computation involving the cut will be
represented by three nodes in the lift.

We first prove the following property, which is the inductive step of the
main proof. 
\begin{quotation}
If $Q_{i-1}=A,Q'$ is a node in the pruned tree $T$, with a ground
instance  $Q_{i-1}\sigma_{i-1}$ such that  $S\models Q_{i-1}\sigma_{i-1}$
then there exists in $T$ a descendant  $Q_{j}$ of  $Q_{i-1}$ with a ground instance 
$Q_{j}\sigma_{j}$, 
such that $S\models Q_{j}\sigma_{j}$.

 Moreover, $Q_{i-1},Q_{j}$ are the first and the last query of  an LD-derivation
 $D_{i,j}$
 (for a $P'\supseteq P$) which  is a lift of  
an unrestricted derivation \cite[Df.\,5.35]{Doets}
beginning with $Q_{i-1}\sigma_{i-1}$ and ending with $Q_{j}\sigma_{j}$.
 Also, the resultant of $D_{i,j}$ is the same
 as the  
 resultant of the derivation  $Q_{i-1},Q_i,\ldots,Q_{j}$ for $P$, which is a
 path in $T$ between  $Q_{i-1}$ and $Q_{j}$.

\end{quotation}
\noindent
If the selected atom  $A\sigma_{i-1}$ of $Q_{i-1}\sigma_{i-1}$
is covered by a clause without the cut
then the proof follows that of Th.\,\ref{th:completeness}
\cite[Th.\,4]{drabent.corr.2012};
$Q_j$ is a child of   $Q_{i-1}$ in $T$.

The main part of this proof deals with the case when
a ground instance  $A\sigma_{i-1}$ of the selected atom
is adjustably covered by a clause $C = H\gets \seq[k-1]A,!,A_k,\ldots,A_n$
(and thus covered by $C$).
Without loss of generality we may assume that $C$ is a clause variant used in
the resolution step.
Let $Q_{i-1}=A,Q'$.
Then $A\in pre$, $A$ is unifiable with $H$ with an mgu $\theta_i$, 
$H\theta_i\in pre$,
atom $A\sigma_{i-1}$ is an instance of $H\theta_i$,
and 
$Q_i=(\seq[k-1]A,!,A_k,\ldots,A_n,Q')\theta_i$ is a child of $Q_{i-1}$
(in $T_0$ and in $T$).
{\sloppy\par}

Some ground instance of $Q_i$ consists of atoms from $S$;
hence $S\models\exists(\seq[k-1]A)\theta_i$. 
As $P$ is semi-complete w.r.t.\ $S$,  the LD-tree $T'$ for
$(\seq[k-1]A)\theta_i$ has an infinite, or a successful branch.
Assume that each successful branch of $T'$ is preceded by an infinite one.
Then, by (\ref{property:pruned:tree}),
 the cut introduced in $Q_i$ is not executed, and 
$T$ contains an infinite branch, contradiction.
So $T'$ has a successful branch not preceded by an infinite one. 
Hence, by (\ref{property:pruned:tree}),
 $T$ contains a path, starting in $Q_i$, of the form
{\sloppy\par}

$$
\begin{array}{l}
Q_i =(\seq[k-1]A,!,A_k,\ldots,A_n,Q')\theta_i,
 \\
Q_{i+1} =
(\ldots,A_2,\ldots,A_{k-1},!,A_k,\ldots,A_n,Q')\theta_i\theta_{i+1},  \\
\cdots \\
Q_j = (A_k,\ldots,A_n,Q')\theta_i\cdots\theta_{j},  \\
\end{array}
$$
where $\theta_{i+1},\ldots,\theta_j$ are the used mgu's,
and each query $Q_i,\ldots,Q_{j-1}$ contains more atoms than $Q_j$ does.

Let $\phi=\theta_{i+1}\cdots\theta_{j}$. We have
$A_1\theta_i\phi,\ldots, A_{k-1}\theta_i\phi\in post$
(hence all ground instances of $A_1\theta_i\phi,\ldots, A_{k-1}\theta_i\phi$
are in $post$).
To apply
condition \ref{def:adjustably:tough:condition} of Df.\,\ref{def:adjustably},
take $\rho = \restrict{\theta_i}{C}$ (the restriction of $\theta_i$ to the
variables of $C$).
Then  $do m(\rho)\subseteq vars(H)$,
$r n g(\rho) \cap vars(C) \subseteq vars(H)$ (as $\theta_i$ is a relevant
unifier of $A$ and $H$),
$do m(\rho)\cap r n g(\rho) = \emptyset$ (as $\theta_i$ is idempotent)
and, obviously, $C\theta_i=C\rho$.
By  condition \ref{def:adjustably:tough:condition} of Df.\,\ref{def:adjustably}
(with $\eta = \restrict\phi{(\seq[{k-1}]A)\rho}$),
$A\sigma_{i-1}$ is covered by 
$(H \gets A_k, \ldots, A_n)\rho\eta$, which is
$(H \gets A_k, \ldots, A_n)\theta_i\phi$.
So $A\sigma_{i-1} = H\theta_i\phi\sigma'$ and
$A_k\theta_i\phi\sigma', \ldots, A_n\theta_i\phi\sigma' \in S$,
for some $\sigma'$.
Let $\psi= \theta_i\phi\sigma'$.

 Let us now have a different look at the derivation $Q_{i-1},\ldots,Q_j$.
Let us introduce a new predicate symbol $p$ and and treat $\seq[k-1]A$ as
terms. 
Consider
\[
\begin{array}{l}
Q_{i-1}=A,Q' \\  
Q_i' =( p(\seq[k-1]A),!,A_k,\ldots,A_n,Q')\theta_i,  \\
Q_j = (A_k,\ldots,A_n,Q')\theta_i\phi.  \\
\end{array}
\]
It 
is a derivation for a program $\{C',C''\}$, where
\[
\begin{array}{l}
C' =  H\gets p(\seq[k-1]A),!,A_k,\ldots,A_n, \\
C'' =  p(\seq[k-1]A)\theta_i\phi.
\end{array}
\]
The mgu's used are $\theta_i$ and $\phi$.
We construct an unrestricted derivation
 $\Gamma$ for $\{C',C''\}$  \cite[Df.\,5.9]{Doets},
so that derivation $\Delta =Q_{i-1}, Q_{i}, Q_{j}$ is a lift
\cite[Df.\,5.35]{Doets} of~$\Gamma$.  
$\Gamma$~consists of ground queries   $R_1, R_2, R_3$, where
\[
\begin{array}{l}
R_1 = Q_{i-1}\sigma_{i-1} = (A,Q')\sigma_{i-1}, \\
R_2 = ( p(\seq[k-1]A),!,A_k,\ldots,A_n)\psi,Q'\sigma_{i-1},  \\
R_3 = ( A_k,\ldots,A_n)\psi,Q'\sigma_{i-1}. \\
\end{array}
\]

So $\Gamma$ is an unrestricted derivation, where $C',C''$ are the applied
clauses.  Hence $\Delta$ is a lift of  $\Gamma$.
By the lifting theorem  \cite[Th.\,5.37]{Doets}, the resultant
$R_3\to R_1$ of $\Gamma$ is an instance of the resultant 
$Q_j\to  Q_{i-1}\theta_i\phi$ of $\Delta$.
    The latter is also the resultant of the original derivation 
    $Q_{i-1},Q_i,Q_{i+1}\ldots,Q_j$ for $P$.
Note that  $S\models R_3$.
So $R_3$ is the required ground instance of $Q_j$.
This completes the proof of the inductive step.

\medskip
By induction we obtain that in $T$ there exists a successful (as $T$ is finite)
branch  $Q,\ldots,\Box$ with a subsequence of nodes
 $\Delta'=Q, Q_{j_1}\ldots,Q_{j_l},\Box$ (where $0<j_1<\cdots<j_l$)
which is a lift of a ground successful unrestricted derivation starting with
$Q\theta$ (for some superset of the program $P$).
The resultants (i.e.\ the answers) for both successful derivations are the
same.  
Hence, by the lifting theorem  \cite[Th.\,5.37]{Doets}
$Q\theta$ is an instance of the answer of $\Delta'$, hence of an answer of $T$.
\hfill $\Box$ \quad \mbox{}
\end{proof}

It remains to show that if condition 
 \ref{def:adjustably:tough:condition} of Df.\,\ref{def:adjustably}
holds for an $H\rho$ then it holds for all its instances (for which $A$ is an
instance) 

\renewcommand\AA{{\ensuremath{\overrightarrow A}}\xspace}

\begin{lemma}
\label{lemma:instances-derivations}
Let 
$C$ be a clause $H\gets \seq[k-1]A,!,A_k,\ldots,A_n$ ($0\leq k\leq n$).
Let \AA be $\seq[{k-1}]A$.  
Let $A$, $\rho,\ \eta$, $H\rho $ i $\AA\rho\eta$ be as in condition 
 \ref{def:adjustably:tough:condition} of Df.\,\ref{def:adjustably}.
The following conditions (1) and (2) are equivalent.

\medskip\noindent
(1) 
$A$ is covered by 
$(H \gets A_k,\ldots,A_n)\rho\eta$ w.r.t.\ $S$.

\medskip\noindent
(2) There exists a successful LD-derivation for $A$ using in its consecutive
steps the clauses $C$,  $ A_1\rho\eta,\ldots,A_{k-1}\rho\eta$, and then
some atoms from $S$.
\end{lemma}
Note that in (2) all the clauses used in the derivation, except $C$, are ground.

\begin{proof}
\mbox{(1) $\Rightarrow$ (2)}:
\hspace{0pt plus .5em}%
(1) implies that 
 $A$ is covered by a ground clause $(H \gets A_k,\ldots,A_n)\rho\eta\sigma$.
Construct
an LD-derivation $D$ for $A$, using first clause $C\rho\eta\sigma$ and then the
clauses as in (2).  Its lift is a required derivation.
{\sloppy\par}

\medskip\noindent
(2) $\Rightarrow$ (1):
Take a derivation as in (2):
\[
\begin{array}l
  A	\\
  (\seq A)\theta_1	\\
  \ldots		\\
  (\SEQ A {k} n)\theta_1\cdots\theta_k	\\  
  \ldots		\\
  A_n\theta_1\cdots\theta_n	\\  
  \Box
\end{array}
\]
with mgu's \seq[n+1]\theta, which are ground substitutions.
We have
$A=H\theta_1 = H \SEQC \theta 1 {n+1}$, 
and the ground clauses used in the derivation are
$A_i\SEQC\theta 1 {i+1} = A_i\SEQC \theta 1 {n+1}$  ($i=1,\ldots,n$). 
Then  $\AA \rho\eta = \AA \SEQC \theta 1 {n+1}$    and 
$A_i\SEQC \theta 1 {n+1}\in S$ for  $i=k,\ldots,n$.

Now $A = H\rho\delta$  for some ground substitution $\delta$ with 
$do m(\delta)=vars(H\rho)$. 
So 
$\theta_1=\restrict{(\rho\delta)}{vars(H)}$, as $do m(\theta_1)=vars(H)$.
Note that $do m(\delta) \cap vars(C)\subseteq vars(H)$
(as $r n g(\rho) \cap vars(C)\subseteq vars(H)$).
Hence
$\theta_1=\restrict{(\rho\delta)}{vars(H)} = 
\restrict{(\rho\delta)}{{\it vars}(C)}$, and thus
$C\theta_1 = C\rho\delta$.
In particular, $\AA\theta_1 = \AA\rho\delta$.
So  $\AA\rho\eta 
 = \AA\SEQC\theta1{n+1}
 = \AA\rho\delta \SEQC\theta2{n+1} 
$.
Thus
$\eta
      = \restrict{(\delta \SEQC \theta 2 {n+1})}{\AA\rho}
$
(as $do m(\eta)=vars(\AA\rho)$).

Let
$\sigma = {(\delta \SEQC \theta 2{n+1})} \setminus \eta$.
As $\delta$ and $\sigma$ are ground and with disjoint domains,
$\delta \SEQC \theta 2{n+1} = \eta\cup\sigma = \eta\sigma$.
Hence
$C\SEQC \theta 1{n+1}= C\rho\delta\SEQC \theta 2{n+1} = C\rho\eta\sigma$
(as $C\theta_1 = C\rho\delta$).
So
$H\rho\eta\sigma= H\SEQC \theta 1{n+1} = A$ and 
$A_i\rho\eta\sigma = A_i\SEQC \theta 1{n+1}\in S$, for $i=k,\ldots,n$.
Hence
$A$ is covered by $(H \gets A_k,\ldots,A_n)\rho\eta$ w.r.t.\ $S$.
\hfill $\Box$
\end{proof}

\begin{lemma}
\label{lemma:instances:adjustably:tough:condition}
If condition \ref{def:adjustably:tough:condition} of Df.\,\ref{def:adjustably}
holds for an atom $H\rho\in pre$ then it holds for any its instance 
$H\rho'$ such that $A$ is an instance of $H\rho'$, and $\rho'$ satisfies the
requirements of condition  \ref{def:adjustably:tough:condition}
(i.e.\
 $do m(\rho')\subseteq vars(H)$, $r n g(\rho') \cap vars(C) \subseteq vars(H)$, 
 $do m(\rho')\cap r n g(\rho') = \emptyset$).
\end{lemma}

\begin{proof}
Let \AA be $\seq[{k-1}]A$.  
We first show that $\AA\rho'$ is an instance of $\AA\rho$.
For some~$\delta$ with $do m(\delta)\subseteq vars(H\rho)$, we have
$H\rho'=H\rho\delta$, so $\rho'=\restrict{(\rho\delta)}{H}$.
Consider a variable $X$ from $C$. There are two cases:
\\
1.  $X\in vars(H)$, thus  $X\rho'=X\rho\delta$. \
\\
2.  $X\not\in vars(H)$. So $X\not\in do m(\rho)$, 
as $do m(\rho) \subseteq vars(H)$.
From $do m(\delta)\subseteq vars(H)\cup r n g(\rho)$ it follows that
 $do m(\delta)\cap vars(C)\subseteq vars(H)$
(as $r n g(\rho) \cap vars(C) \subseteq vars(H)$).
So $X\not\in do m(\delta)$.  Hence $X\rho\delta=X$ and $X\rho'=X\rho\delta$.

   We showed that $\rho' = \restrict{(\rho\delta)}C$.
So $\AA\rho'=\AA\rho\delta $.
Then each ground instance $\AA\rho'\eta'$ of $\AA\rho'$  
such that $ A_1\rho'\eta',\ldots,A_{k-1}\rho'\eta'\in post$
is an instance of $\AA\rho$
($\AA\rho'\eta'= \AA\rho\delta\eta'= \AA\rho\eta$ where 
 $\eta=(\restrict{\delta}{ \mbox{\scriptsize$\AA$}\rho})\eta'$).

Assume that condition  \ref{def:adjustably:tough:condition} holds for $H\rho$.
Then for each ground instance as above where each atom of 
$\AA\rho\eta$ is in $post$,
atom $A$ is covered w.r.t.\ $S$ by  $(H \gets A_k,\ldots,A_n)\rho\eta$.
By Lemma \ref{lemma:instances-derivations} 
there exists a successful LD-derivation for $A$ using in its consecutive
steps the clauses $C$,  $ A_1\rho\eta,\ldots,A_{k-1}\rho\eta$, and then
some atoms from $S$. As $ A_i\rho\eta= A_i\rho'\eta'$ for $i=1,\ldots,k-1$,
by Lemma \ref{lemma:instances-derivations} used in the opposite direction,
$A$ is covered by  $(H \gets A_k,\ldots,A_n)\rho'\eta'$.
\hfill $\Box$
\end{proof}

\bibliographystyle{plain}
\bibliography{bibshorter,bibpearl,bibmagic}
\pdfbookmark[1]{References}{bookmark:bib}

\end{document}